\documentclass[10pt]{article}

\usepackage{amssymb} 
\usepackage{latexsym} 
\usepackage{amsmath}
\usepackage{color}
\usepackage{epsfig}
\usepackage[mathcal]{eucal}
\usepackage{enumerate}
\usepackage{anysize}
\usepackage[english]{babel}
\usepackage{bbm}
\usepackage{hyperref}

\newtheorem{theorem}{Theorem}
\newtheorem{corollary}[theorem]{Corollary}

\newenvironment{proof}{\noindent{\bf Proof.}}%
               {\hspace*{\fill}$\Box$\par\vspace{4mm}}

\def\cadre{$$\vcenter\bgroup\advance\hsize by -2em\noindent
             \refstepcounter{equation}(\theequation)~\ignorespaces}
\makeatletter
\def\endcadre{\egroup\eqno$$\global\@ignoretrue}
\def\ncadre{$$\vcenter\bgroup\advance\hsize by -2em\noindent
             \ignorespaces}
\makeatletter
\def\endncadre{\egroup\eqno$$\global\@ignoretrue}

\newcommand{\comment}[1]{}
\newcommand{\mybreak} {\par\vspace{2mm}\noindent}

\newcommand{\introd}          {\mbox{introd}}

\makeatletter
\def\imod#1{\allowbreak\mkern10mu({\operator@font mod}\,\,#1)}
\makeatother

\newcommand{\B} {\mathcal{B}}
\newcommand{\LL} {\mathcal{L}}
\newcommand{\LLC} {\mathcal{L}^\circ}

\newcommand{\Hasse}{\mathbf{H}}
\newcommand{\HasseC}{\mathbf{H}^\circ}
\newcommand{\maxbic} {\B}

\newcommand{\twin} {\mbox{twin}}
\newcommand{\pend} {\mbox{pend}}

\newcommand{\algofast}{\texttt{FastComputeBDHDiagram}}
\newcommand{\algo}{\texttt{ComputeBDHDiagram}}

\newcounter{progcount}
\newcounter{linecount}[progcount]

\newcommand{\N}{\refstepcounter{linecount}\thelinecount. \>}
\newcommand{\NL}[1]{\refstepcounter{linecount}\thelinecount. \label{#1}\>}

\newenvironment{prog}[1]{
    \refstepcounter{progcount}\label{#1}
    \par\vspace{0.5ex}\noindent\hspace{1ex}
    \begin{minipage}{\linewidth}
    \small
    \begin{tabbing}
    =spa\=spa\=spa\=spa\=spa\=spa\=spa\=spa\=spa\=spa\=spa\=spa\=\kill
}%
{
    \end{tabbing}
    \end{minipage}\\[0.5ex]
}

\newcommand{\key}[1]{\textbf{#1~}}\ignorespaces

\begin{document}
\pagestyle{plain}

\title{On Computing the Galois Lattice of Bipartite Distance Hereditary Graphs}

\author{Nicola Apollonio\footnote{Istituto per le Applicazioni del
Calcolo, M. Picone, v. dei Taurini 19, 00185 Roma, Italy.
\texttt{nicola.apollonio@cnr.it}} \and
{Paolo Giulio Franciosa\footnote{Dipartimento di Scienze Statistiche,
Sapienza Universit\`a di Roma, p.le Aldo Moro 5, 00185
Roma, Italy. \texttt{paolo.franciosa@uniroma1.it}. Partially supported by the Italian Ministry of Education,
University, and Research (MIUR) under PRIN 2012C4E3KT national
research project ``AMANDA -- Algorithmics for MAssive and Networked
DAta''.}}
}

%

\maketitle

\begin{abstract} The class of Bipartite Distance Hereditary (BDH) graphs is the intersection between bipartite domino-free and chordal bipartite graphs.\ Graphs in both the latter classes have linearly many maximal bicliques, implying the existence of polynomial-time algorithms for computing the associated Galois lattice.\ Such a lattice can indeed be built in $O(m\times n)$ worst case-time for a domino-free graph with $m$ edges and $n$ vertices.\ In this paper we give a sharp estimate on the number of the maximal bicliques of BDH graphs and exploit such result to give an $O(m)$ worst case time algorithm for computing the Galois lattice of BDH graphs. By relying on the fact that neighborhoods of vertices of BDH graphs can be realized as directed paths in a arborescence, we give an $O(n)$ worst-case space and time encoding of both the input graph and its Galois lattice, provided that the reverse of a Bandelt and Mulder building sequence is given.
\\
\end{abstract}


\section{Introduction}\label{sec:intro}
Enumerating inclusion-wise maximal vertex-sets of complete bipartite subgraphs (maximal bicliques) in bipartite graphs is a challenging theoretical and computational problem~\cite{Eppstein,AACFHS,MakinoUno}  related to several classical problems in combinatorial optimization, theoretical computer
science~\cite{Amil,CornazFon,Agarwal,Berry07} and bioinformatics~\cite{Liliuliwong,Zhang} (and the references cited therein). The problem has been shown to be \#P-complete by Kuztnetsov~\cite{kuz} and there have been active efforts to bound and estimate the number of maximal bicliques as well as efficiently computing and listing such bicliques both in general and in restricted classes of bipartite graphs~\cite{prisner,Amil}. There are two non-trivial classes of bipartite graphs admitting polynomially many maximal bicliques: the class of bipartite domino-free graphs~\cite{Amil} and the class of $C_6$-free graphs~\cite{prisner} (in particular, the class of chordal-bipartite graphs): $O(m)$ in the former case, $m$ being the size of the graph, and $O((n_1\times n_2)^2)$ in the latter case, $n_1$ and $n_2$ being the number of vertices in the two color classes. In these cases, the interest is clearly on designing efficient algorithms to count the number of maximal bicliques, list all the maximal bicliques, and solving related computations. However, besides its own interest, what makes the problem even more appealing even in special cases, is the intimate relationship with the problem of building \emph{concept lattices} (also known as Galois lattices) of a formal context in Formal Concept Analysis, a well established (though still flourishing) topic in Applied Lattice Theory~\cite{gw}.

For our purposes, a formal context is a bipartite graph $G$ with color classes $X$ and $Y$. In Formal Concept Analysis, $X$ is interpreted as a set of objects and $Y$ as a set of attributes, while $G$ encodes the incidence binary relation between attributes and objects: object $x\in X$ has attribute $y\in Y$ if and only if $xy$ is an edge of $G$. A formal concept is an ordered pair $(X_0,Y_0)$, where $X_0$ is a subset of objects, $Y_0$ is a subset of attributes, and all the objects in $X_0$ share all the attributes in $Y_0$ in such a way that any other object $x\in X\setminus X_0$ fails to have at least one of the attributes in $Y_0$ and any other attribute $y\in Y\setminus Y_0$ is not possessed by at least one object in $X_0$. The sets $X_0$ and $Y_0$ are called the intent and the extent of the formal concept. Concepts can be (partially) ordered from the more specific to the more general: the more objects share a common set of attributes the less specific is the concept, e.g.\ ``mammal'' is less specific than ``dog'', the extent of the concept ``mammal'' contains the extent of ``dog'' as well as the extent of ``cats'' for instance, and dually the intent of ``mammal'', namely the set of attributes defining ``mammal'', is contained in the intent of ``dog''. It is convenient to assume the existence of the most specif concept of a context, namely the concept whose intent is the set of all attributes, as well as the most general concept, namely the concept whose extent consists of all objects. As proved by Ganter and Wille, according to the basic Theorem of concept lattices, the set of formal concepts of a given context, hierarchically ordered, is actually a lattice called the \emph{concept lattice} of the context $G$ (also known as the Galois lattice of $G$), with the most specific and the most general concepts as bottom and top, respectively.

From a graph-theoretical point of view, formal concepts can be identified with the maximal bicliques $B$ of $G$, hence if $\B(G)$ denotes the collection of the maximal bicliques of $G$, then $\mathcal{L}(G)=(\maxbic(G)\cup\{\bot,\top\},\preceq)$ is the Galois lattice of $G$, where, $\bot$ and $\top$ are two dummy maximal bicliques consisting, respectively, of the color class $X$ alone and the color class $Y$ alone (unless there are universal vertices in $G$) and the partial order $\preceq$ is defined by
$$B\preceq B' \Leftrightarrow X\cap B\subseteq X\cap B'.$$
Equivalently, the same partial order can be defined as
$$B\preceq B' \Leftrightarrow Y\cap B\supseteq Y\cap B'$$
since $X\cap B\subseteq X\cap B' \Leftrightarrow Y\cap B\supseteq Y\cap B'$ for any pair of maximal bicliques $B$ and $B'$.

Hence, with any bipartite graph there is an associated lattice on its collection of maximal bicliques and the shape of such a lattice can be characteristic of particular classes of bipartite graphs. 

For instance \emph{Bipartite Distance Hereditary} graphs (BDH for shortness) have been investigated in \cite{acfDAM,acfIWOCA}. Recall that a graph is \emph{Distance Hereditary} if the distance between any two of its vertices is the same in every connected induced subgraph containing them. A graph is \emph{Bipartite Distance Hereditary} if it is both bipartite and distance hereditary.

In \cite{acfDAM}, BDH graphs have been characterized as the class of bipartite graphs whose Galois lattice is tree-like. More precisely, it has been shown that the Hasse diagram $\HasseC(G)$ of the poset obtained by removing the top and bottom elements from the Galois lattice $\LL(G)$ of a bipartite graph $G$ is a tree if and only if $G$ is a BDH graph. This implies that the linear dimension of the Galois lattice of a BDH graph is at most 2. Anyway, no efficient algorithms for computing the Galois lattice of a BDH graph have been proposed, though special classes of graphs inducing efficiently computable Galois lattices (much more efficiently than in the general case) have been investigated \cite{Amil,BerryMSS06}. In particular, an $O(m\times n)$ worst case-time algorithm  has been given in \cite{Amil}  for computing the Galois lattice for the more general class of domino-free graphs with $m$ edges and $n$ vertices.
\mybreak
Bandelt and Mulder \cite{bm} proved that BDH graphs are exactly all the graphs that can be constructed starting from a single vertex by a sequence of adding
pending vertices and false twins of existing vertices. This is a special case of what happens for (not necessarily bipartite) distance hereditary graphs, that are characterized as  graphs that can be built starting from a single vertex and a sequence of additions of pending vertices, false twins and true twins (see Section~\ref{sec:prelim} for the definition of false and true twins). This sequence is referred to as an \emph{admissible sequence} in~\cite{bm}, and it is the reverse of what is called a \emph{pruning sequence} in \cite{hammer}. Damiand \emph{et al.}~\cite{damiand} proposed
an optimal $O(m)$ worst case time algorithm for computing a pruning sequence of a distance hereditary graph $G$, where $m$ is the number of edges in $G$, using a cograph recognition algorithm in~\cite{corneil}. Obviously, the same algorithm computes a pruning sequence of a BDH graph.
\mybreak
In this paper we show that, for any BDH graph $G$ with $n$ vertices and $m$ edges:
\begin{itemize}
\item $G$ contains at most $n-2$ maximal bicliques. This improves, for BDH graphs, the more general $O(m)$ bound given in~\cite{Amil} for domino-free bipartite graphs and the $O(n^4)$ bound in~\cite{prisner} for $C_6$-free graphs;

\item the total size of $\LL(G)$, i.e., the sum of the number of vertices over all maximal bicliques of $G$, is $O(m)$.

\item it is possible to compute $\HasseC(G)$, i.e., the Hasse diagram of the Galois lattice of $G$ in worst case time $O(m)$. This improves by a factor of $n$ (the number of vertices of $G$) the $O(m\times n)$ worst case time algorithm given in \cite{Amil} for the larger class of domino-free graphs.\ The construction we propose also finds meet-irreducible and join-irreducible elements in the Galois lattice, also known as \emph{introducers} (see next section), and provides an explicit representation of all maximal bicliques in $G$.
This result is based on a simpler constructive proof that BDH graphs have a tree-like Galois lattice (i.e., the \emph{if part} of the characterization in~\cite{acfDAM}).

\item it is possible to compute an arborescence $A$ such that $G$ is the arcs/paths incidence graph of a set of paths in $A$. The arborescence can be computed in worst case time $O(n)$, starting from a pruning sequence of $G$, and gives an $O(n)$ space representation of both neighborhoods in $G$ and maximal bicliques. This result provides a simpler and constructive proof of the maximal bicliques encoding proposed in~\cite{acfDAM};

\item relying on the arborescence representation above, it is possible to compute an $O(n)$ space representation of  $\HasseC(G)$.\ The compact representation is obtained in $O(n)$ time starting from a pruning sequence of $G$, yielding an overall $O(m+n)$ algorithm to compute $\HasseC(G)$.
\end{itemize}

\section{Definitions and preliminaries}\label{sec:prelim}
Graphs dealt with in this paper are simple (no loops nor parallel edges).\ The neighborhood in $G=(V,E)$ of a vertex $v \in V$ is the set $N_G(v) = \{u\ |\ uv \in E\}$, and the number of vertices in $N_G(v)$ is denoted by $\deg_G(v)$ and it is called the \emph{degree of $v$ in $G$}.
A vertex $v$ is said to be a \emph{pending vertex in $G$} if $\deg_G(v) = 1$.
A vertex $v$ is said to be a \emph{false twin} \emph{in $G$} if a vertex $u \not= v$ exists such that $N_G(v) = N_G(u)$ and $v \not\in N_G(u)$, while $v$ is said to be a \emph{true twin} \emph{in $G$} if a vertex $u \not= v$ exists such that $N_G(v) = N_G(u)$ and $v \in N_G(u)$. Since we are dealing with bipartite graphs, and bipartite graphs cannot contain true twins, we will refer to false twins simply by \emph{twins}.
For ease of notation, we usually omit the subscript referring to $G$ when no confusion can arise. Occasionally, we denote the edge-set (arc-set) of a (directed) graph $G$ by $E(G)$.

An \emph{arborescence} is a directed tree with a
single special node distinguished as the \emph{root} such that,
for each other vertex, there is a dipath from the root to that
vertex. An arborescence $T$ induces a partial order $\leqslant_T$ on $E(T)$, the \emph{arborescence order}, as follows: $e\leqslant_T f$ if the unique path from the root of $T$ which ends with $f$ contains $e$. So we can think of $T$ as the partially ordered set $(E(T),\leqslant_T)$. The arborescence order allows us to identify paths $F$ of $T$ with intervals of the form $[\alpha(F),\beta(F)]$, where $\alpha(F)$ is the arc of $F$ closest to the root of $T$ (the $\leqslant_T$-least element of $F$) and $\beta(F)$ is the arc of $F$ farthest from the root of $T$ (the $\leqslant_T$-greatest element of $F$).

The color classes of a bipartite graph $G$ are referred to as the \emph{shores} of $G$. If the bipartite graph $G$ has shores $X$ and $Y$, we denote such a graph by $G=(X,Y,E)$. A \emph{complete bipartite graph} is a bipartite graph $(X,Y,E)$ where edge $xy \in E$ for each $x \in X$, $y\in Y$.
A \emph{biclique} $B$ in $G$ is a set of vertices of $G$ that induces a complete bipartite subgraph $(X',Y',E')$ with $X' \not= \emptyset$ and $Y' \not= \emptyset$. Such a biclique will be identified with the pair $(X',Y')$ of the shores of the graph it induces, and the shores of a biclique $B$ will be denoted by $X(B)$ and $Y(B)$. A biclique in $G$ is a \emph{maximal biclique} if it is not properly contained in any biclique of $G$.
\mybreak
The \emph{transitive reduction} of a partially ordered set $(S,\leqslant)$ is the directed acyclic graph on
$S$ where there is an arc leaving $x\in S$ and entering $y\in S$ if and
only if $x\leqslant y$ and there is no $z\in S \setminus \{x,y\}$ such that $x\leqslant z\leqslant y$.  With some abuse of terminology, we refer to the transitive reduction of a partially ordered set as to its \emph{Hasse diagram}. Arcs $(x,y)$ in the Hasse diagram of $(S,\leqslant)$, will be denoted by $x\dot < y$ and will be referred to as \emph{cover pairs}. If $x\dot < y$ is a cover pair in the Hasse diagram of some poset we say that $y$ \emph{covers} $x$ or that $x$ \emph{is covered} by $y$.
\mybreak
As above, the set of maximal bicliques in $G$ is denoted by $\B(G)$, and $\mathcal{L}(G)=(\maxbic(G)\cup\{\bot,\top\},\preceq)$ is the associated Galois lattice. Moreover, we let $\mathcal{L}^\circ(G)=(\maxbic(G),\preceq)$, namely, $\mathcal{L}^\circ(G)$ is the partially ordered set obtained from $\mathcal{L}(G)$ after removing $\bot$ and $\top$. The symbols $\Hasse(G)$ and $\HasseC(G)$ denote the Hasse diagrams of $\mathcal{L}(G)$ and $\mathcal{L}^\circ(G)$, respectively.
\mybreak
It is known (see~\cite{gw}) that for any vertex $v \in X$ (resp., $v \in Y$) of a bipartite graph $G=(X,Y,E)$ there is a maximal biclique in $G$ (hence an element in $\B(G)$) of the form $(\bigcap_{x \in N(v)}N(x), N(v))$ (resp., $(N(v), \bigcap_{x \in N(v)}N(x))$). Conforming to concept lattice terminology, such an element of $\LL(G)$ is referred to as \emph{object concept} (resp., \emph{attribute concept}), while it is called the \emph{introducer} of $v$ in~\cite{Berry13}. The introducer of $v$ is the lowest (resp., highest) maximal biclique containing $v$ in $\LL(G)$, and will be denoted by $\introd(v)$. It can also be shown that irreducible elements in $\LL(G)$ are introducers---recall that in a partially ordered set (in particular in a lattice)  an element $r$ is \emph{meet-irreducible} (resp., \emph{join-irreducible}) if $r$ is not the least upper bound (resp., the greatest lower bound) of any two other distinct elements $s$ and $t$-- however, we do not use such notions here. 
\mybreak
Given a BDH graph $G$, we assume that the reverse of a pruning sequence for $G$ has been computed as in~\cite{bm}, for example applying the $O(m)$ algorithm in~\cite{damiand}. Hence, we know that $G$ can be built starting from a sigle vertex $v_1$, and adding a sequence of pending vertices and twin vertices $v_2, v_3, \ldots, v_n$. For $1 \leqslant i \leqslant n$, we denote by $G_i$ the subgraph of $G$ induced by $v_1, v_2, \ldots, v_i$.
The neighborhood of a vertex $v_j$ in $G_i$, for $j \leqslant i$, is denoted by $N_i(v_j)$, and the degree of $v_j$ in $G_i$ is denoted by $\deg_i(v_j) = |N_i(v_j)|$.
The number of maximal bicliques in $G_i$ containing vertex $v$ is denoted by $b_i(v)$. 

Actually, the reverse of the pruning sequence in \cite{damiand,hammer} is defined for (not necessarily bipartite) distance hereditary graphs, and consists in a sequence $S = [s_2, s_3, \ldots, s_{n}]$ of triples, where $s_i = (v_i, C_i, v_k)$, and the value $C_i \in \{P,F,T\}$, for $k < i$, specifies whether $v_i$ is a pending vertex ($P$) of $v_k$ in $G_i$, or $v_i$ is a false twin ($F$) of $v_k$ in $G_i$, or $v_i$ is a true twin ($T$) of $v_k$ in $G_i$. In the case of bipartite distance hereditary graphs, the pruning sequence contains only pending vertices and false twins.

So, for $2 \leqslant i \leqslant n$, either $|N_i(v_i)| = 1$ (i.e., $v_i$ is a pending vertex in $G_i$), or a vertex $v_k$ exists, with $k<i$, such that $N_i(v_i) = N_i(v_k)$ (i.e., $v_i$ is a false twin of $v_k$ in $G_i$).


\section{Incremental construction of the Galois lattice of a BDH graph}

We can now describe how the Galois lattice of a BDH graph evolves during the Bandelt and Mulder construction. The following result holds for general bipartite graphs.
\begin{theorem}\label{th:induction}
If a bipartite graph $G_{i}=(X_i,Y_i,E_i)$ is obtained from a bipartite graph $G_{i-1}$ by adding a twin of an existing vertex or a pending vertex, then either $\HasseC(G_{i})$ is isomorphic to $\HasseC(G_{i-1})$ or $\HasseC(G_i)$ is obtained from $\HasseC(G_{i-1})$ by adding a pending vertex.
\end{theorem}
\begin{proof}
We distinguish two cases, depending on the added vertex $v_i$ being a pending vertex or a twin vertex. Let us initially assume $v_i \in X_i$.
\begin{description}
\item[$v_i$ is a twin vertex:] let $v_i$ be a twin of $v_k$ in $G_{i}$, with $k < i$. Since $N_i(v_i) = N_i(v_k)$, for each maximal biclique $(X',Y')$ in $G_{i-1}$, with $v_k \in X'$, there is a maximal biclique $(X' \cup \{v_i\},Y')$ in $G_{i}$.
Maximal bicliques in $G_{i-1}$ not containing $v_k$ remain unchanged in $G_i$. This changes do not alter the order relation among bicliques. Hence, $\HasseC(G_{i})$ is isomorphic to $\HasseC(G_{i-1})$.
\item[$v_i$ is a pending vertex:] let $v_i$ be a pending vertex of $v_j$, so $v_j \in Y$ and $N_{i}(v_j)=N_{i-1}(v_j) \cup \{v_i\}$. The only maximal biclique in $G_i$ containing $v_i$ is $(N_{i}(v_j), \{v_j\})$.
We distinguish two cases: either $(N_{i-1}(v_j), \{v_j\})$ is a maximal biclique in $G_{i-1}$, or not.
\begin{itemize}
\item $(N_{i-1}(v_j), \{v_j\})$ is a maximal biclique in $G_{i-1}$: the maximal biclique $(N_{i-1}(v_j), \{v_j\})$ in $\LLC(G_{i-1})$ is replaced in $\LLC(G_{i})$ by the maximal biclique $(N_{i-1}(v_j) \cup \{v_i\}, \{v_j\})$. So, no new maximal bicliques are created and the order relation among existing maximal bicliques is unchanged. Hence, $\HasseC(G_{i})$ is isomorphic to $\HasseC(G_{i-1})$.

\item $(N_{i-1}(v_j), \{v_j\})$ is not a maximal biclique in $G_{i-1}$:  $\LLC(G_i)$ is obtained from $\LLC(G_{i-1})$ by adding the maximal biclique $B = (N_{i-1}(v_j) \cup \{v_i\}, \{v_j\})$ and a cover pair $B' \dot\prec B$, where $B'$ is the greatest maximal biclique containing $v_j$ in $\LLC(G_{i-1})$. The new biclique $B$ is the introducer of $v_j$ in $\LLC(G_{i})$, while $B'$ is the introducer of $v_j$ in $\LLC(G_{i-1})$. It is immediate to see that $B' \dot\prec B$ is the only new cover pair in $\LLC(G_i)$ with respect to $\LLC(G_{i-1})$. Hence, $\HasseC(G_{i})$ is obtained by adding a pending vertex to $\HasseC(G_{i-1})$, which is a maximal element in $\LLC(G_i)$.

\end{itemize}
\end{description}
\noindent
In case $v_i \in Y_i$, the only differences in the above arguments consist in swapping the shores in the maximal bicliques, and, as for the latter case, in adding a new cover pair $B \dot\prec B'$ (instead of $B' \dot\prec B$)---thus  a new minimal element is added to $\LLC(G_i)$ instead of a new maximal element.
\end{proof}

Theorem~\ref{th:induction} provides the induction step to prove that the Galois lattice of a BDH graph is a tree.
\begin{corollary}\label{co:isatree}
If $G$ is a BDH graph, then $\HasseC(G)$ is a tree.
\end{corollary}
\begin{proof}
Any BDH graph $G$ can be built by a sequence of vertex additions as in Theorem~\ref{th:induction}, starting from a single vertex. Graph $G_3$ is isomorphic to $K_{1,2}$, hence $G_3$ contains only one maximal biclique, and $\HasseC(G_3)$ is a tree consisting in a single vertex. After each addition of pending vertices or twins of existing vertices, $\HasseC(G_{i-1})$ is turned into $\HasseC(G_{i})$ which is either isomorphic to $\HasseC(G_{i-1})$ or is obtained from $\HasseC(G_{i-1})$ by the addition of a pending vertex by Theorem~\ref{th:induction}. Therefore $\HasseC(G_{i})$ is either the same tree as $\HasseC(G_{i-1})$ or it is obtained from a tree by adding a pending vertex. Since adding a pending vertex to a tree always results in a tree the thesis follows. 
\end{proof}

Note that Corollary~\ref{co:isatree} provides a simpler and constructive proof of the \emph{if part} of Theorem 1 in \cite{acfDAM}.

\section{Bounding the size of the Galois lattice}

Another consequence of Theorem~\ref{th:induction} is that BDH graphs have few maximal bicliques.
\begin{corollary}\label{co:linearbiclique}
The number of maximal bicliques in a BDH graph on $n$ vertices is at most $n - 2$.
\end{corollary}
\begin{proof}
Graph $G_3$ contains only one maximal biclique, since it is isomorphic to $K_{1,2}$.
By Theorem~\ref{th:induction}, the number of maximal bicliques in $G_i$ is at most the number of maximal bicliques in $G_{i-1}$ plus one, for $4 \leqslant i \leqslant n$.
\end{proof}
Corollary~\ref{co:linearbiclique} shows that the number of maximal bicliques of a BDH graph with $n$ vertices and $m$ edges is always smaller than $n$.\ The class of BDH graphs is the intersection of bipartite domino-free graphs and bipartite $C_{2k}$-free graphs, $k\geqslant 3$, namely \emph{chordal bipartite graphs}. The best known bound for the number of maximal bicliques for both classes is $O(m)$ (see~\cite{Amil,KloKra}).
Moreover, we can bound the number of maximal bicliques containing a given vertex.
\begin{theorem}\label{th:grado}
Each vertex $v$ of a BDH graph $G$ is contained in at most $2 \cdot \deg(v) - 1$ maximal bicliques, where $\deg(v)$ is the degree of $v$ in $G$.
\end{theorem}
\begin{proof}
Assume without loss of generality that $v \in X$, and let $\mathcal{Y}_v=\left(Y(B) \ |\ v\in X(B),\, B\in \maxbic\right)$ be the family of the $Y$-shores of the maximal bicliques containing $v$. Each member of $\mathcal{Y}_v$ is a subset of $N(v)$. We show that $\mathcal{Y}_v$ is a \emph{laminar family}, namely it has the property that for any two members $Y_1$ and $Y_2$ either such two members are disjoint or one is included in the other. Since it is well-known (see \cite[Chapter 2.2]{korte}) that a laminar family consisting of subsets of a common ground set of $k$ elements contains at most $2k-1$ sets, the thesis follows once we prove that $\mathcal{Y}_v$ is indeed laminar. Suppose to the contrary that there are $Y_1$ and $Y_2$ in $\mathcal{Y}_v$ such that all the following conditions hold: $Y_1 \cap Y_2 \not= \emptyset$, $Y_1 \not\subseteq Y_2$ and $Y_2 \not\subseteq Y_1$. Hence there are maximal bicliques $B_0$, $B_1$, $B_2$ and $B_3$ such that $Y(B_1)=Y_1$, $Y(B_2)=Y_2$, $B_0\leqslant B_i$, $i=1,2$ and $B_i\leqslant B_3$, $i=1,2$: just choose for $B_0$ the introducer of $v$ and for $B_3$ the smallest maximal biclique such that $Y(B_3)\supseteq  Y_1\cap Y_2$. If one picks $v_1 \in X_1 \setminus X_2$, $v_2 \in X_2 \setminus X_1$, $w \in Y_1 \cap Y_2$, $w_1 \in Y_1 \setminus Y_2$, $w_2 \in Y_2 \setminus Y_1$, then $\{v,v_1,v_2,w,w_1,w_2\}$ induces a domino in $G$, contradicting that $G$ is a BDH graph.
\end{proof}
In view of Theorem~\ref{th:grado}, we can bound the total size of the Galois lattice of a BDH graph $G$, i.e., the sum of the number of vertices in each maximal biclique in $\LL(G)$.
\begin{corollary}\label{co:size}
The Galois lattice of a BDH graph $G$ has total size $O(m)$, where $m$ is the number of edges in $G$.\end{corollary}
\begin{proof}
Let $X$ and $Y$ be the shores of $G$. By Theorem~\ref{th:grado}, each vertex $v$ in $X$ appears in at most $2\cdot\deg_G(v)-1$ maximal bicliques. Therefore $\sum_{B \in \B(G)} |X(B)| \leqslant 2m-n$. Analogously, $\sum_{B \in \B(G)} |Y(B)| \leqslant  2m-n$.
\end{proof}

\section{An $O(m)$ algorithm for computing the Galois lattice of a BDH graph}
The computation of the Galois lattice of a BDH graph $G$ starts from the reverse $v_1, v_2, \ldots, v_n$ of a pruning sequence for $G$ (following the terminology in \cite{damiand,hammer}). The pruning sequence can be computed in $O(m)$ time for general (not necessarily bipartite) distance hereditary graphs, as shown by Damiand \emph{et al.} in \cite{damiand}, where the authors provide a fix for a previous algorithm presented by Hammer and Maffray \cite{hammer}.
The basic ideas to compute the Hasse diagram $\HasseC(G)$ of $\LLC(G)$, starting from the pruning sequence, are given in the proof of Theorem~\ref{th:induction}. We describe here how that approach leads to an $O(m)$ time algorithm.
Note that, thanks to Corollaries~\ref{co:isatree} and \ref{co:size}, an explicit description of $\HasseC(G)$, containing an exhaustive listing of the set of vertices in each maximal biclique and all cover pairs, can be given in $O(m)$ space. 

The algorithm is incremental, i.e., for each $1\leqslant i \leqslant n$, the Hasse diagram $\HasseC(G_i)$ of the graph $G_i$ induced by $v_1, v_2, \ldots, v_i$ is computed by updating $\HasseC(G_{i-1})$. Note that each $G_i$ is a BDH graph as well.
Our algorithm also computes, for each vertex $v \in X$ (resp., $v \in Y$), its introducer $\introd(v)$, i.e., the lowest (resp., highest) maximal biclique in $\HasseC(G)$ containing $v$. So, it is possible to retrieve all the $p$ maximal bicliques containing $v$ in time $O(p)$, by means of a simple upwards (resp., downwards) traversal in the tree-like $\HasseC(G)$ starting from $\introd(v)$.

The algorithm is shown in Figure~\ref{fi:algom}.
For each maximal biclique $B=(X(B), Y(B))$ in $\HasseC(G_i)$, we maintain the following information:
\begin{itemize}
\item the list of vertices in $X(B)$;
\item the list of vertices in $Y(B)$;
\item the list of maximal bicliques covered by $B$ in $\HasseC(G_i)$;
\item the list of maximal bicliques that cover $B$ in $\HasseC(G_i)$.
\end{itemize}
Moreover, for each vertex $v$ in $G_i$ we store a reference to $\introd(v)$ in $\HasseC(G_i)$.

In case a twin vertex $v_i$ of $v_k$ is added, $v_i$ behaves exactly in the same way as $v_k$, so we just add $v_i$ to all the maximal bicliques containing $v_k$. In order to retrieve all these maximal bicliques we start from $\introd(v_k)$ and follow all upward arcs (if $v_i \in X$) or all downward arcs (if $v_i \in Y$). We also set $\introd(v_i)$ to $\introd(v_k)$.

In case a pending vertex $v_i$ of $v_k$ is added, then two cases may occur, depending on $B = (N_i(v_k), \{v_k\})$ (resp., $B =(\{v_k\}, N_i(v_k))$ if $v_i \in Y$) being a maximal biclique in $\HasseC(G_{i-1})$ or not.
In case $B$ is a maximal biclique in $\HasseC(G_{i-1})$, then we just add $v_i$ to $B$, and set $\introd (v_i)$ to $\introd(v_k)$.
Otherwise, if $B$ is not a maximal biclique in $\HasseC(G_{i-1})$, then $\HasseC(G_i)$ contains one more maximal biclique with respect to $\HasseC(G_{i-1})$, the biclique $(N_i(v_k), \{v_k\})$ (or $(\{v_k\}, N_i(v_k))$), which also is the introducer of $v_i$ and $v_k$  in $\HasseC(G_i)$.

\begin{figure}[ht]
\noindent\hrulefill%
\begin{prog}{pr:simplehasse}
Given:\\
-- a BDH graph $G$,\\
-- the reverse of a pruning sequence for $G$ $(v_2, C_2, v_{k_2}), (v_3, C_3, v_{k_3}), \ldots, (v_n, C_n, v_{k_n})$, where $C_i \in \{P, T\}$,\\
compute $\HasseC(G)$.\\
W.l.o.g., we assume $v_1 \in X$ and $v_2 \in Y$\\
\vspace{0.3cm}\\
\N  $H \leftarrow$ a single biclique $(\{v_1\},\{v_2\})$\\
\N  $\introd(v_1) \leftarrow (\{v_1\},\{v_2\})$\\
\N  $\introd(v_2) \leftarrow (\{v_1\},\{v_2\})$\\
\\
\N  \key{for}$i=3$ to $n$\\
  /*\\
  we assume $v_i \in X$. Changes in case $v_i \in Y$ are straightforward,\\
  except the change in Line \ref{li:x1} (see Line~\ref{li:y1})\\
  */\\
\N  \>  \key{if}$v_i$ is a twin vertex in $G_i$\\

\N  \>  \>  \key{let}$v_k$ be the twin vertex of  $v_i$ in $G_i$\\
\N  \>  \>  \key{let}$B = \introd(v_k)$\\
\N  \>  \>  \key{add}$v_i$ to $X(B)$\\
\NL{li:traversal}  \>  \>  \key{for each}maximal biclique $B'$ in $H$ such that $B \prec B'$\\
\N  \>  \>  \>  \key{add}$v_i$ to $X(B')$\\
\N  \>  \>  $\introd(v_i) \leftarrow \introd(v_k)$\\

\N  \>  \key{else} /* $v_i$ is a pending vertex in $G_i$ */\\

\N  \>  \>  \key{let}$v_k$ be the vertex $v_{k_i}$ adjacent to $v_i$\\
\N  \>  \>  \key{let}$B = \introd(v_k)$\\
\NL{li:test}  \>  \>  \key{if}$B = (N_{i-1}(v_k), \{v_k\})$\\
\N  \>  \>  \>  \key{add}$v_i$ to $X(B)$\\
\N  \>  \>  \>  $\introd(v_i) \leftarrow \introd(v_k)$\\
\N  \>  \>  \key{else} /* $\introd(v_k) \not= (N_{i-1}(v_k), \{v_k\})$ */\\
\N  \>  \>  \>  \key{create}a new maximal biclique $B' = (N_i(v_k), \{v_k\})$\\

\NL{li:x1}  \>  \>  \>  \key{add}$B'$ to $H$ so that $B \prec B'$\\
\NL{li:y1}  \>  \>  \>  /* in case $v_i \in Y$: \key{add}$B'$ to $H$ so that $B' \prec B$ */\\
\N  \>  \>  \>  $\introd(v_k) = B'$\\
\N  \>  \>  \>  $\introd(v_i) = B'$\\

\N  \key{end for}\\
\N  \key{return}$H$\\
\end{prog}

\caption{Algorithm \algo.}
\protect\label{fi:algom}
\noindent\hrulefill%
\end{figure}

\begin{theorem}
Algorithm \algo\ requires $O(m)$ worst case time.
\end{theorem}
\begin{proof}
The test in Line~\ref{li:test} can be performed in constant time, by just checking whether the appropriate
shore in $B$ has size one.
The loop in Line~\ref{li:traversal} is performed by a simple traversal of $H$, requiring overall linear time in the number of maximal bicliques in which $v_i$ must be added, since the traversed portion of $H$ is a tree. Thus, the overall time complexity is given by the number of vertices that are added to each maximal biclique, which is $O(m)$ by Corollary~\ref{co:size}.
\end{proof}


\section{Compact representation of neighborhoods and maximal bicliques}\label{se:arbo}

A BDH graph may contain up to $\Theta(n^2)$ edges---for example, any complete bipartite graph is a BDH graph. Anyway, the neighborhood of each vertex can be conveniently encoded by a compact representation. The following theorem, proved in~\cite{acfDAM}, shows that a BDH graph always is the incidence graph of arcs of an arborescence and a set of paths in the arborescence.
\begin{theorem}[see Apollonio \emph{et al.} \cite{acfDAM}, Theorem 5]\label{th:arbo}
Let $G$ be a BDH graph with shores $X$ and $Y$. There exist arborescences $T_X$ and $T_Y$ and bijections $\psi_X: X\rightarrow E(T_X)$ and $\psi_Y:Y\rightarrow E(T_Y)$ such that $\psi_X(N(y))$ is a directed path in $T_X$ for each $y \in Y$ and $\psi_Y(N(x))$ is a directed path in $T_Y$ for each $x \in X$.
\end{theorem}
By Theorem~\ref{th:arbo}, it is possible to store a pair of arborescences $T_X$ and $T_Y$ so that the neighborhood of each vertex in $X$ (resp., $Y$) is implicitly represented by the two extremes of a dipath in $T_Y$ (resp., $T_X$). Such an implicit representation still requires overall $O(n)$ worst case space and
generalizes to an arborescence what is possible on a path for \emph{convex bipartite graphs}, namely, bipartite graphs $G=(X,Y,E)$ for which a linear order $L$ on $Y$ exists such that $N(x)$ is an interval in $L$ for each $x \in X$. For convex bipartite graphs, each neighborhood can be implicitly represented by the extremes of the corresponding interval in $L$. Nonetheless, BDH graphs are not convex, neither they are $c$-convex (in the sense of~\cite{albano}) in general, even for small $c \in \mathbb{N}$.
\mybreak
The pair of arborescences mentioned in Theorem~\ref{th:arbo}, along with the corresponding bijections, can be computed by specializing the algorithm of Swaminathan and Wagner~\cite{SwaWA} that runs Bixby and Wagner's
algorithm~\cite{BixWA} for the \emph{Graph Realization Problem}---roughly: the recognition problem for graphic matroids--as a subroutine. The ensuing running time is $O(\alpha(\nu, m)\times m)$ where $\nu$ is the number of vertices in shore $Z$, $Z\in \{X,Y\}$, while $m$ is the number of edges of $G$ and  
$\alpha(\cdot,\cdot)$ is a function which grows very slowly and behaves essentially as a constant even for
large values of both its arguments.
\mybreak
We give here a simpler constructive proof of Theorem~\ref{th:arbo}, that also provides a much simpler and more efficient (though much less general) algorithm to compute the arborescence representation.

\vspace{2mm}
\begin{proof} \textbf{(of Theorem~\ref{th:arbo}, constructive)}
Since the role of the shores of $G$ is symmetrical, it suffices to prove the existence of an arborescence $T_Y$ and a bijection $\psi_Y$ fulfilling the thesis.\ The proof is carried out by induction on graphs $G_i$'s in the Bandelt and Mulder construction sequence of $G$. Let $X_i$ and $Y_i$, be the shores of $G_i$. For ease of notation we set $T_i=T_{Y_i}$ and $\psi=\psi_{Y_i}$. Hence we identify $\psi_{Y_i}$ with the restriction of $\psi$ on the vertices of $Y_i$.
We assume, w.l.o.g., that $v_1 \in X$ and $v_2 \in Y$. 
\mybreak
Graph $G_2$ is necessarily isomorphic to $K_{1,1}$.\ Thus the thesis trivially holds for $G_2$: $T_2$ consists of a single arc $e$, with $\psi(v_2) = e$.\ The neighborhood $Y_2$ of the unique vertex in $X_2$ is mapped into a path consisting of the unique arc $e \in E(T_2)$.
\mybreak
Assume the thesis holds for $G_{i-1}$, with $i > 2$. The neighborhood $N_{i-1}(v_k)$, for each $v_k \in  X_{i-1}$, is mapped by $\psi$ into a dipath in $T_{i-1}$.
When adding vertex $v_i$ we distinguish four cases, since $v_i$ can be added either to shore $X$ or to shore
$Y$, and it can be either a pending vertex or a twin vertex of an existing vertex.

\begin{enumerate}[(i)]
\item $v_{i}$ is a twin vertex in shore $X$: let $v_j\in X$ be a twin of $v_i$. The arborescence is unchanged.  Since $N_i(v_{i}) = N_i(v_j) = N_{i-1}(v_j)$, and $\psi(N_{i-1}(v_j))$ was a path in $T_{i-1}$, then $\psi (N_i(v_i))$ is a path in $T_i$.

\item\label{case:twinY} $v_{i}$ is a twin vertex in shore $Y$: let $v_j\in y$ be a twin of $v_i$. We subdivide arc $\psi(v_j)$ into two consecutive arcs $\psi(v_i)$ and $\psi(v_j)$, by adding a new vertex to the arborescence. If $v_j \in N_i(x)$ for some $x \in X_{i}$, then also $v_i \in N_i(x)$, hence $\psi(N_i(x))$ contains both arc $\psi(v_j)$ and arc $\psi(v_i)$. Any path containing $\psi(v_j)$ is thus extended to a path containing $\psi(v_i)$. Therefore, $\psi(N_i(x))$ is still a path in $T_i$, for each $x \in X_i$.

\item $v_{i}$ is a pending vertex in shore $X$: the arborescence is unchanged. The neighborhood $N_i(v_i)$ is a single vertex, so $\psi (N_i(v_i))$ is a path consisting of a single arc.

\item\label{case:pendingY} $v_{i}$ is a pending vertex in shore $Y$: let $N_i(v_{i}) = \{v_j\}$. Only the neighborhood of $v_j$ is changed, with $N_i(v_j) = N_{i-1}(v_j) \cup \{v_i\}$. We add a new vertex and a new arc $\psi(v_i) = e$ to the arborescence, so that arc $e$ is adjacent to the last arc in the path $\psi(N_{i-1}(v_j))$. Since $\psi(N_{i-1}(v_k))$ is a path in $T_{i-1}$, for $v_k \in X_{i-1} \setminus \{v_j\}$, then $\psi(N_i(v_k))$ is a path in $T_i$. Moreover, $\psi(N_{i-1}(v_j))$ is a path in $T_{i-1}$ as well; hence $\psi(N_i(v_j))$ is a path in $T_i$, consisting of the concatenation of $\psi(N_{i-1}(v_j))$ and $e$.
\end{enumerate}
The only cases in which arcs are added to $T_{i-1}$ are~(\ref{case:twinY}) and~(\ref{case:pendingY}). It is immediate to see that, in both cases, if $T_{i-1}$ is an arborescence then also $T_i$ is an arborescence.
\end{proof}
An example of the above construction is shown in Figure~\ref{fig:arborescence}.
\begin{figure}[h]
\noindent\hrulefill%
    \begin{center}
    	\def\svgwidth{15cm}
             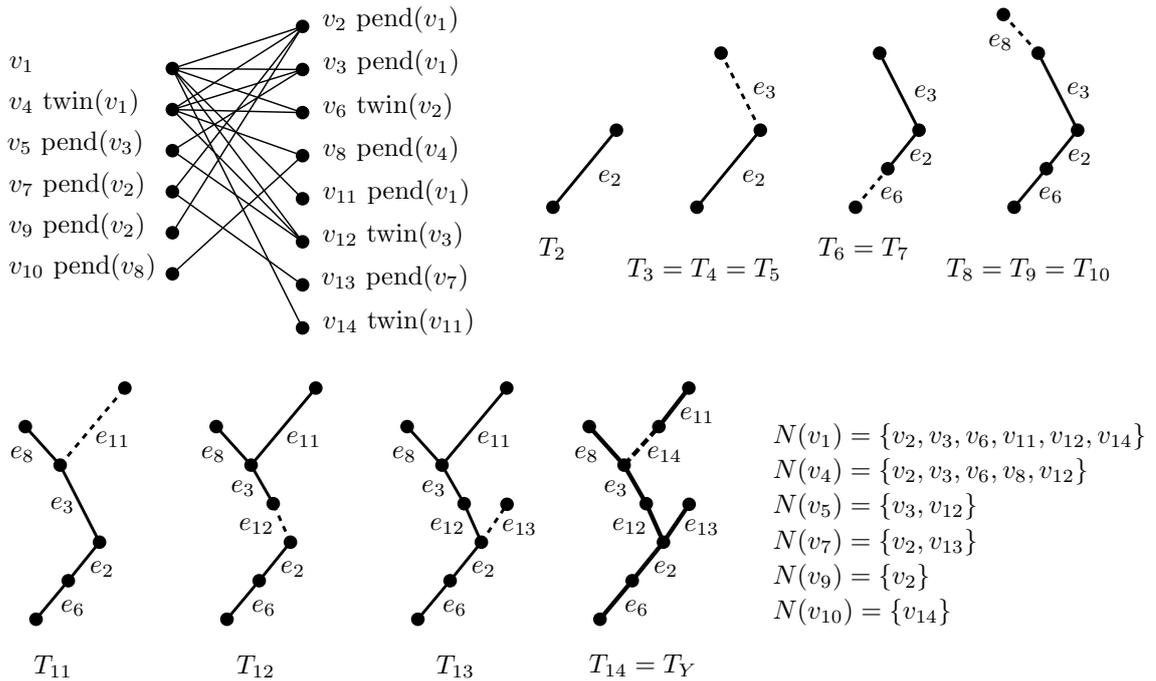
    \caption{A BDH graph and its supporting arborescence $T_Y$, where $Y$ is the shore on the right. The arborescence $T_Y$ is obtained incrementally under the addition in the graph of pending vertices and twin vertices $v_1, v_2, \ldots, v_{14}$, as described in the proof of Theorem \ref{th:arbo}.\comment{, in the order $A, 1, 2, B, C, 3, D, 4, E, F, 5, 6, 7, 8$} Labels $\pend(v)$ and $\twin(v)$ in the graph denote the insertion of a pending vertex adjacent to $v$ or a twin vertex of $v$. Arc $\psi(v_i)$ in the arborescence is labeled by $e_i$. Dashed arcs are the arcs added to each arborescence. Observe that the neighborhood of each vertex in $Y$ is mapped to a dipath in $T_Y$. For example, $N(v_7)$ is mapped to the dipath from $\alpha(N(v_7)) = e_{2}$ to $\beta(N(v_7)) = e_{13}$, while $N(v_4)$ is mapped to the dipath from $\alpha(N(v_4)) = e_{6}$ to $\beta(N(v_4)) = e_{8}$.}
    \protect\label{fig:arborescence}
    \end{center}
\noindent\hrulefill%
\end{figure}
\mybreak
Since each shore of a maximal biclique is an intersection of neighborhoods, and the intersection of dipaths in an arborescence is itself a dipath, we can encode each maximal biclique $(X(B),Y(B))$ of a BDH graph by means of a dipath in $T_X$ and a dipath in $T_Y$. It is therefore convenient to introduce a unique map $\psi:X \cup Y\rightarrow E(T_X) \cup E(T_Y)$ as follows: $\psi(v) = \psi_X(v)$ if $v \in X$ and  $\psi(v) = \psi_Y(v)$ if $v \in Y$. The following fact follows now straightforwardly.
\begin{corollary}\label{co:arbo}
Given a BDH graph $G$ with shores $X$ and $Y$, there exist two arborescences $T_X$ and $T_Y$, and a bijection $\psi:X \cup Y\rightarrow E(T_X) \cup E(T_Y)$, such that for each maximal biclique $B\in \LLC(G)$ we have that $\psi(X(B))$ is a directed path in $T_X$ and $\psi(Y(B))$ is a directed path in $T_Y$.
\end{corollary}
Let $G$ a BDH graph with shores $X$ and $Y$ and let $S\in\left(N(x),\, x\in X\right)\cup\left(N(y),\, y\in Y\right)$. Let $T\in \{T_X,T_Y\}$. Dipaths in $T$ are identified by intervals in the arborescence order induced by $T$ (recall Section~\ref{sec:prelim}).\ In particular, dipath $\psi(S)$ of $T$ is identified with $[\alpha(\psi(S)),\beta(\psi(S))]$.\ For ease of notation we set $\alpha(\psi(S))=\alpha(S)$ and $\beta(\psi(S))=\beta(S)$.\ Hence, a maximal biclique $B$ will be encoded by the two intervals $[\alpha(X(B)), \beta(X(B))]$ in $T_X$ and $[\alpha(Y(B)), \beta(Y(B))]$ in $T_Y$.

\section{An $O(n)$ time algorithm for computing a compact representation of the Galois lattice of a BDH graph}
The arborescence representation described in Theorem~\ref{th:arbo} and Corollary~\ref{co:arbo}, together with the $O(n)$ upper bound in Corollary~\ref{co:linearbiclique} on the number of
maximal bicliques in $\LLC(G)$, allows us to derive an $O(n)$ space encoding of the Galois lattice of a BDH graph. We show here how this encoding can be computed in $O(n)$ worst case time. An exhaustive listing of the $k$ vertices in each maximal biclique can still be obtained in optimal $O(k)$ time by traversing the compact representation.

Algorithm \algofast\ is listed in Figure~\ref{fi:fastcomputehasse}. Starting from the reverse of a pruning sequence of a BDH graph $G$, it computes the two supporting arborescences $T_X, T_Y$
in Theorem~\ref{th:arbo} and an implicit representation of $\HasseC(G)$.
Each maximal biclique $B=(X(B), Y(B))$ in $\HasseC(G)$ is implicitly represented by the two intervals
$[\alpha(X(B)), \beta(X(B))]$ and $[\alpha(Y(B)), \beta(Y(B))]$. The list of the $k$ arcs in an interval can be retrieved in $O(k)$ time by a simple walk in the arborescence, starting from $\beta(\cdot)$ and following parent pointers to $\alpha(\cdot)$.
Thus, the Hasse diagram of the Galois lattice can be represented in $O(n)$ space, also including the two arborescences needed to list vertices in maximal bicliques when required.

The algorithm we propose also computes, for each vertex $v \in X$ (resp., $v \in Y$), its introducer. This allows us to retrieve all the $p$ maximal bicliques containing $v$ in time $O(p)$.
 
Algorithm \algofast\ follows the same steps as Algorithm \algo\ but, when a vertex in the reverse of the pruning sequence is processed, in addition to updating $\HasseC(G_{i-1})$ to $\HasseC(G_i)$, also the two arborescences $T_X$ and $T_Y$ are updated according to the proof of Theorem~\ref{th:arbo}.

For each maximal biclique $B=(X(B), Y(B))$ in $\HasseC(G_i)$, we maintain the following information:
\begin{itemize}
\item the set of vertices in $X(B)$, represented through the end-arcs $\alpha(X(B)),\beta(X(B))$ of the associated dipath in $T_X$;
\item the set of vertices in $Y(B)$, represented through the end-arcs $\alpha(Y(B)),\beta(Y(B))$ of the associated dipath in $T_Y$;
\item the list of maximal bicliques covered by $B$ in $\HasseC(G_i)$;
\item the list of maximal bicliques that cover $B$ in $\HasseC(G_i)$.
\end{itemize}
Moreover, for each vertex $v_k$ in $G_i$, we store a reference to $\introd(v_k)$ in $\HasseC(G_i)$.

In the algorithm we only show how to process pending vertices and twin vertices in $X$, the algorithm and the data structures being completely symmetric with respect to swapping shore $X$ for shore $Y$.

\begin{figure}[ht]
\noindent\hrulefill%
\begin{prog}{pr:fasthasse}
given a BDH graph $G$ and the reverse of a pruning sequence\\
for $G$ $[(v_2, C_2, v_{k_2}), (v_3, C_3, v_{k_3}), \ldots, (v_n, C_n, v_{k_n})],$\\
compute:\\
-- $\HasseC(G)$\\
-- the arborescences $T_X$ and $T_Y$ representing vertices in $X$ and $Y$ as in Theorem~\ref{th:arbo}\\
~~~in which $\psi(v_i)$ is the arc denoted by $e_i$\\
\vspace{0.3cm}\\
for each vertex $v$ in $G_i$ we maintain a reference to $\introd(v)$ in $\HasseC(G_i)$;\\
for each maximal biclique $B = (X(B), Y(B)) \in \HasseC(G_i)$ we maintain:\\
-- the list of maximal bicliques covered by $B$ in $\HasseC(G_i)$;\\
-- the list of maximal bicliques that cover $B$ in $\HasseC(G_i)$;\\
-- the end-arcs $\alpha(X(B)), \beta(X(B))$ of the dipath in $T_X$ representing $X(B)$;\\
-- the end-arcs $\alpha(Y(B)), \beta(Y(B))$ of the dipath in $T_Y$ representing $Y(B)$.\\
\vspace{0.3cm}\\
/* w.l.o.g., we assume $v_1 \in X$ and $v_2 \in Y$ */\\
\N  $H \leftarrow$ a single biclique $B = (\{v_1\},\{v_2\})$\\
\N  $\introd(v_1) \leftarrow B$\\
\N  $\introd(v_2) \leftarrow B$\\
\N $T_X \leftarrow \mbox{a single arc } e_1$\\
\N $T_Y \leftarrow \mbox{a single arc } e_2$\\
\N  $\alpha(X(B)) \leftarrow e_1$; $\beta(X(B)) \leftarrow e_1$\\
\N  $\alpha(Y(B)) \leftarrow e_2$; $\beta(Y(B)) \leftarrow e_2$\\
\\
\N  \key{for}$i=3$ to $n$\\
  /*\\
  for the sake of simplicity we assume $v_i \in X$\\
  changes in case $v_i \in Y$ are straightforward, except the change in Line \ref{li:x} (see Line~\ref{li:y})\\
  */\\
\N  \>  \key{if}$v_i$ is a twin vertex in $G_i$\\

\N  \>  \>  \key{let}$v_k$ be the twin vertex of  $v_i$ in $G_i$\\
\N  \>  \>  \key{update}$T_X$ by splitting $e_k$ into two arcs $e_i, e_k$, with $e_i \prec_{T_X} e_k$\\
\NL{li:inizioblocco}  \>  \>  \key{for each}maximal biclique $B'=(X(B'), Y(B'))$ in $H$ with $\alpha(X(B')) = e_k$\\
\NL{li:fineblocco}  \>  \>  \>  $\alpha(X(B')) \leftarrow e_i$\\
\N  \>  \>  $\introd(v_i) \leftarrow \introd(v_k)$\\

\N  \>  \key{else} /* $v_i$ is a pending vertex in $G_i$ */\\

\N  \>  \>  \key{let}$v_k$ be the vertex adjacent to $v_i$\\
\N  \>  \>  \key{append}a new arc $e_i$ in $T_X$ as a leaf above $\beta(X(\introd(v_k)))$\\ 
\N  \>  \>  \key{if} $\introd(v_k) = (N_{i-1}(v_k), v_k)$ /* i.e., $\alpha(Y(\introd(v_k))) = \beta(Y(\introd(v_k)))$ */\\
\N  \>  \>  \>  $\beta(X(\introd(v_k))) \leftarrow e_i$\\ 
\N  \>  \>  \>  $\introd(v_i) \leftarrow \introd(v_k)$\\
\N  \>  \>  \key{else} /* $\introd(v_k) \not= (N_{i-1}(v_k), v_k)$ */\\
\N  \>  \>  \>  \key{create}a new maximal biclique $B'$\\
\NL{li:x}  \>  \>  \>  \key{add}$B'$ to $H$ as a leaf so that $\introd(v_k) \prec B'$\\
\NL{li:y}  \>  \>  \>  /* in case $v_i \in Y$: \key{add}a leaf $(v_k,N_{i}(v_k))$ so that $(v_k,N_{i}(v_k)) \prec \introd(v_k)$ */\\
\N  \>  \>  \>  $\alpha(X(B')) \leftarrow \alpha(X(\introd(v_k)))$ /* because $B' = (N_{i-1}(v_k) \cup \{v_i\}, v_k)$ */\\
\N  \>  \>  \>  $\beta(X(B')) \leftarrow e_i$\\
\N  \>  \>  \>  $\alpha(Y(B')) \leftarrow e_k$\\
\N  \>  \>  \>  $\beta(Y(B')) \leftarrow e_k$\\
\N  \>  \>  \>  $\introd(v_k) \leftarrow (N_{i}(v_k), v_k)$\\
\N  \>  \>  \>  $\introd(v_i) \leftarrow (N_{i}(v_k), v_k)$\\

\N  \key{end for}\\
\N  \key{return}$H$\\
\end{prog}

\caption{Algorithm \algofast.}
\protect\label{fi:fastcomputehasse}
\noindent\hrulefill%
\end{figure}

\begin{theorem}
Starting from the pruning sequence of a BDH graph $G$ on $n$ vertices, an implicit representation of its Galois lattice can be computed in $O(n)$ worst case time and space. Retrieving the $p$ vertices in each maximal biclique requires $O(p)$ worst case time, and retrieving the $k$ maximal bicliques containing a given vertex requires $O(k)$ worst case time.
\end{theorem}
\begin{proof}
It is immediate to see that each step in Algorithm \algofast,  except for the loop in Lines~\ref{li:inizioblocco} and \ref{li:fineblocco}, needs constant time per vertex to update the two arborescences $T_X$, $T_Y$. This because the Hasse diagram $\HasseC(G)$ contains, for each maximal biclique $B=(X', Y')$, its implicit representation $\alpha(X')$, $\beta(X')$, $\alpha(Y')$ and $\beta(Y')$.
Concerning Lines~\ref{li:inizioblocco} and \ref{li:fineblocco}, instead of updating the value of $\alpha(X(B'))$ for each maximal biclique in $H$ with $\alpha(X(B')) = e_k$, we store a single reference to the $\alpha(\cdot)$ value for all maximal bicliques sharing the same value of $\alpha(\cdot)$ (analogously for $\alpha(Y(B'))$), thus the set of updates in Lines~\ref{li:inizioblocco} and \ref{li:fineblocco} can be performed in constant time by just substituting that reference.

The set of vertices in the $X$ shore of a maximal biclique $B$ can be listed by traversing $T_X$ starting from $\beta(X(B))$, following parent pointers, until $\alpha(X(B))$ is reached,
and analogously for the $Y$ shore on $T_Y$. The set of maximal bicliques containing vertex $v \in X$ (resp., $v \in Y$) can be reached by traversing $\HasseC(G)$ upward (resp., downward) starting from $\introd(v)$. Since $\HasseC(G)$ is a tree, each maximal biclique is reached only once during the traversal.
\end{proof}

\bibliographystyle{plain}
\bibliography{BDHlattice}
\end{document}